\ifx\DCGVersion\undef
   \newcommand{\DCGVer}[1]{}
   \newcommand{\NotDCGVer}[1]{#1}
\else
   \newcommand{\DCGVer}[1]{#1}
   \newcommand{\NotDCGVer}[1]{}
\fi

\NotDCGVer{%
   \documentclass[12pt]{article}%
}
\DCGVer{%
   \documentclass[smallextended,envcountsame]{svjour3} 
}

\DCGVer{\RequirePackage{fix-cm}}%

\NotDCGVer{%
   \usepackage[cm]{fullpage}%
}

\DCGVer{%
   \smartqed  
}

\usepackage{lmodern}

\usepackage[T1]{fontenc}

\usepackage{graphicx}%
\usepackage{mathtools}%
\usepackage{amsmath}%
\usepackage{xcolor}%
\usepackage{xspace}%
\usepackage{paralist}%
\usepackage{amssymb}
\usepackage{caption}%
\usepackage{mleftright}%
\usepackage[normalem]{ulem}

\NotDCGVer{%
   \usepackage{titlesec}%
   \titlelabel{\thetitle. }%
}

\usepackage{amsmath}%
\NotDCGVer{%
   \usepackage[amsmath,thmmarks]{ntheorem}%
   \theoremseparator{.}%
}
\usepackage{hyperref}%
\hypersetup{%
   breaklinks,%
   colorlinks=true,%
   linkcolor=[rgb]{0.45,0.0,0.0},%
   citecolor=[rgb]{0,0,0.45} }

\numberwithin{figure}{section}%
\numberwithin{table}{section}%
\numberwithin{equation}{section}%

\NotDCGVer{%
   \theoremstyle{theorem} \newtheorem{theorem}{Theorem}[section]
   \newtheorem{lemma}[theorem]{Lemma}%
   \newtheorem{corollary}[theorem]{Corollary}

   \theoremstyle{remark}%
   \theoremheaderfont{\sf} \theorembodyfont{\upshape}%
   
   \newtheorem{defn}[theorem]{Definition}
   
   \newtheorem*{remark:unnumbered}[theorem]{Remark}%
   \newtheorem{remark}[theorem]{Remark}%
   
   \newtheorem{observation}[theorem]{Observation}
   
}

\DCGVer{%
   \spnewtheorem{defn}[theorem]{Definition}{\itshape}{\rmfamily}
   \spnewtheorem{observation}[theorem]{Observation}{\itshape}{\rmfamily}

}

\newcommand{\myqedsymbol}{\rule{2mm}{2mm}}
 
\newcommand{\si}[1]{#1}

\newcommand{\query}{{x}}
\newcommand{\weight}{\omega}%

\newcommand{\VorC}{\mathcal{V}}%
\newcommand{\VorX}[2][\!]{\VorC\pth{#2}}%

\newcommand{\WVorC}{\mathcal{W}}%
\newcommand{\WVorX}[2][\!]{\WVorC\pth{#2}}%

\newcommand{\VCell}{C}

\newcommand{\VCellPrefix}[1]{K_{#1}}

\newcommand{\face}{F}

\newcommand{\Event}{\mathcal{E}}

\newcommand{\pnt}{{p}}
\newcommand{\pntA}{{q}}
\newcommand{\pntB}{{z}}

\newcommand{\func}{f}

\newcommand{\dist}[2]{\left\| #1 - #2 \right\| }

\newcommand{\PntSet}{{P}}

\providecommand{\ds}{\displaystyle}

\newcommand{\SegSet}{{S}}%
\newcommand{\candid}{{L}}
\newcommand{\candidSet}{\candid}
\newcommand{\candidY}[2]{\candid\pth{#1, #2}}

\newcommand{\flirting}{candidate\xspace}%
\newcommand{\Flirting}{Candidate\xspace}%

\newcommand{\overlayX}[1]{\eta\pth{#1}}

\newcommand{\dsFunc}[2]{\lambda_{#1}\pth{#2}}

\newcommand{\HLinkShort}[2]{\hyperref[#2]{#1\ref*{#2}}}
\newcommand{\HLink}[2]{\hyperref[#2]{#1~\ref*{#2}}}
\newcommand{\HLinkPage}[2]{\hyperref[#2]{#1~\ref*{#2}%
      $_\text{p\pageref{#2}}$}}
\newcommand{\HLinkPageOnly}[1]{\hyperref[#1]{Page~\refpage*{#1}%
      $_\text{p\pageref{#1}}$}}

\newcommand{\HLinkSuffix}[3]{\hyperref[#2]{#1\ref*{#2}{#3}}}
\newcommand{\HLinkPageSuffix}[3]{\hyperref[#2]{#1\ref*{#2}%
      #3$_\text{p\pageref{#2}}$}}

\newcommand{\figlab}[1]{\label{fig:#1}}
\newcommand{\figref}[1]{\HLink{Figure}{fig:#1}}

\newcommand{\seclab}[1]{\label{sec:#1}}
\newcommand{\secref}[1]{\HLink{Section}{sec:#1}}

\newcommand{\corlab}[1]{\label{cor:#1}}
\newcommand{\corref}[1]{\HLink{Corollary}{cor:#1}}%

\newcommand{\apndlab}[1]{\label{apnd:#1}}
\newcommand{\apndref}[1]{\HLink{Appendix}{apnd:#1}}

\newcommand{\remlab}[1]{\label{rem:#1}}
\newcommand{\remref}[1]{\HLink{Remark}{rem:#1}}%

\newcommand{\lemlab}[1]{\label{lemma:#1}}
\newcommand{\lemref}[1]{\HLink{Lemma}{lemma:#1}}%
\newcommand{\lemrefpage}[1]{\HLinkPage{Lemma}{lemma:#1}}

\newcommand{\obslab}[1]{\label{observation:#1}}
\newcommand{\obsref}[1]{\HLink{Observation}{observation:#1}}

\newcommand{\thmlab}[1]{{\label{theo:#1}}}
\newcommand{\thmref}[1]{\HLink{Theorem}{theo:#1}}

\providecommand{\eqlab}[1]{}%
\renewcommand{\eqlab}[1]{\label{equation:#1}}
\newcommand{\Eqref}[1]{\HLinkSuffix{Eq.~(}{equation:#1}{)}}

\newcommand{\etal}{\textit{et~al.}\xspace}

\newcommand{\brc}[1]{\left\{ {#1} \right\}}
\newcommand{\cardin}[1]{\left| {#1} \right|}

\newcommand{\Set}[2]{\left\{ #1 \;\middle\vert\; #2 \right\}}

\newcommand{\pth}[1]{\mleft({#1}\mright)}

\newcommand{\ceil}[1]{\left\lceil {#1} \right\rceil}
\newcommand{\floor}[1]{\left\lfloor {#1} \right\rfloor}
\newcommand{\pbrcx}[1]{\left[ {#1} \right]}
\newcommand{\ExChar}{\mathop{\mathbf{E}}}%
\newcommand{\Ex}[2][\!]{\mathop{\mathbf{E}}#1\pbrcx{#2}}
\newcommand{\ProbChar}{\mathop{\mathbf{Pr}}}
\newcommand{\Prob}[1]{\mathop{\mathbf{Pr}}\!\pbrcx{#1}}
\newcommand{\sep}[1]{\,\left|\, {#1} \bigr.\right.}
\newcommand{\permut}[1]{\left\langle {#1} \right\rangle}%

\newcommand{\nfrac}[2]{#1/#2}%

\newcommand{\intersections}{\psi}%

\newcommand{\object}{\tau}
\newcommand{\SiteSet}{{S}}
\newcommand{\OSiteSet}{{T}}

\newcommand{\site}{{s}}
\newcommand{\siteA}{{r}}

\newcommand{\bisector}{\beta}

\newcommand{\sub}{{H}}
\newcommand{\worst}{g}
\newcommand{\worstX}[1]{g\pth{#1}}

\newcommand{\ElemSet}{{S}}
\newcommand{\RSample}{{R}}

\newcommand{\DefSet}[1]{D\pth{#1}}
\newcommand{\KillSet}[1]{\kappa\pth{#1}}

\newcommand{\Arr}{\mathcal{A}}

\renewcommand{\th}{\si{th}\xspace}

\newcommand{\eps}{\varepsilon}%
\renewcommand{\Re}{{\rm I\!\hspace{-0.025em} R}}

\definecolor{blue25}{rgb}{0,0,0.55}%
\newcommand{\emphic}[2]{%
   \textcolor{blue25}{%
      \textbf{\emph{#1}}}%
   \index{#2}}

\newcommand{\emphi}[1]{\emphic{#1}{#1}}

\newcommand{\cellPL}[2]{%
   {\mathop{{\WVorC}_{\mathrm{cell}}}\pth{#1,#2}}}
\newcommand{\VorCell}[3][\!]{\VorC_\mathrm{cell}\pth{#2, #3}}

\newcommand{\cell}{C}

\newcommand{\LE}[2][\!]{{{\mathcal{L}}}\pth{#2}}

\newcommand{\distSet}[2]{\mathsf{d}\pth{#1, #2}}

\newcommand{\distY}[2]{\left\| {#1} - {#2} \right\|}
\newcommand{\atgen}{\symbol{'100}} \newcommand{\SarielThanks}[1]{%
   \thanks{%
      Department of Computer Science; %
      University of Illinois; %
      201 N. Goodwin Avenue; %
      Urbana, IL, 61801, USA; %
      {\tt \si{sariel}\atgen{}\si{uiuc.edu}}; %
      {\tt \url{http://sarielhp.org}.}%
      #1%
   }%
}
\newcommand{\BenThanks}[1]{\thanks{%
      Department of Computer Science; %
      University of Illinois; %
      201 N. Goodwin Avenue; %
      Urbana, IL, 61801, USA; %
      {\tt \si{raichel}2\atgen{}\si{uiuc}.\si{edu}}; %
      {\tt \url{\si{http://illinois.edu/\string~\si{raichel2}}}}%
      . %
      #1}}

\newcommand{\RIC}{\textsf{RIC}\xspace}
\newcommand{\OSuffix}[1]{\OSiteSet_{#1}^{n}}
\newcommand{\Suffix}[1]{\SiteSet_{#1}^n}
\providecommand{\CNFX}[1]{ {\em{\textrm{(#1)}}}}%
\providecommand{\CNFSoCG}{\CNFX{SoCG}}%
 \newcommand{\pdf}{\xi}
\newcommand{\ObjX}[1]{\mathcal{T}\pth{#1}}
\newcommand{\FuncSet}{\mathcal{F}}%
\newcommand{\OFuncSet}{\mathcal{G}}%
\newcommand{\SetOf}[1]{\mathrm{set}\pth{#1}}

\newcommand{\remove}[1]{}
\newcommand{\edge}{e}%
\newcommand{\arcsX}[2][\!]{\mathrm{arcs}\pth{#2}}%

\newcommand{\ArrVDX}[1]{\Arr_{#1}^{||}}

\NotDCGVer{%
\theoremheaderfont{\em}%
\theorembodyfont{\upshape}%
\theoremstyle{nonumberplain}
\theoremseparator{}
\theoremsymbol{\myqedsymbol}%
\newtheorem{proof}{Proof:}%
}

\allowdisplaybreaks

\begin{document}

\title{On the Complexity of Randomly Weighted Multiplicative Voronoi
   Diagrams%
   \NotDCGVer{%
      \footnote{%
         Work on this paper was partially supported by NSF AF awards
         CCF-0915984, CCF-1217462, and CCF-1421231. 
         A preliminary version of this paper appeared in SoCG 2014
         \cite{hr-ocrwv-14}.%
      }%
   }
}

\author{%
   Sariel Har-Peled%
   \NotDCGVer{\SarielThanks{}}%
   \and%
   Benjamin Raichel%
   \NotDCGVer{\BenThanks{}}%
}

\DCGVer{%
   \institute{%
      S. Har-Peled%
      \and B. Raichel%
      \at Department of Computer Science; University of Illinois; 201
      N.
      Goodwin Avenue; Urbana, IL, 61801, USA.
      %
   }%
}

\date{\today}

\maketitle

\begin{abstract}
    We provide an $O(n \operatorname{polylog} n)$ bound on the
    expected complexity of the randomly weighted multiplicative
    Voronoi diagram of a set of $n$ sites in the plane, where the
    sites can be either points, interior-disjoint convex sets, or
    other more general objects.  Here the randomness is on the weight
    of the sites, not their location. This compares favorably with the
    worst case complexity of these diagrams, which is quadratic.  As a
    consequence we get an alternative proof to that of Agarwal
    \etal~\cite{ahks-urmsn-14} of the near linear complexity of the
    union of randomly expanded disjoint segments or convex sets (with
    an improved bound on the latter).  The technique we develop is
    elegant and should be applicable to other problems.
\end{abstract}


\section{Introduction}

One of the fundamental structures in Computational Geometry is the
\emph{Voronoi diagram} \cite{a-vdsfg-91,akl-vddt-13}; that is, for a
set of points $\PntSet$ in the plane, called sites, partition the
plane into cells such that each cell is the locus of all the points in
the plane whose nearest neighbor is a specific site in $\PntSet$.  In
the plane, the standard Voronoi diagram has linear combinatorial
complexity, but in higher dimensions the complexity is
$\Theta\pth{n^{\ceil{d/2}}}$.  Many generalizations of this
fundamental structure have been considered, including%
\begin{inparaenum}[(i)]
    \item adding weights,
    \item sites that are regions other than points,
    \item extensions to higher dimensions,
    \item other underlying metrics, and
    \item many others.
\end{inparaenum}

Even in the plane, some of these generalizations of Voronoi diagrams
lose their attractiveness as their complexity becomes quadratic in the
worst case.  However, as is often the case, constructions that realize
the quadratic complexity (of say, the weighted multiplicative Voronoi
diagram in the plane) are somewhat contrived, and brittle -- little
changes in the weight dramatically reduces the overall complexity.  To
quantify this observation, we consider here the expected complexity
rather than the worst case of such diagrams, where weights are being
assigned randomly.

\paragraph{Generalizations of Voronoi diagrams.}

In the \emph{additive weighted Voronoi diagram}, the distance to a
Voronoi site is the regular Euclidean distance plus some constant
(which depends on the site).  Additive Voronoi diagrams have linear
descriptive complexity in the plane, as their cells are star shaped
(and thus simply connected), as can be easily verified. This holds
even if the sites are arbitrary convex sets.  In the
\emph{multiplicative weighted Voronoi diagram}, for each site one
multiplies the Euclidean distance by a constant (again, that depends
on the site).  However, unlike the additive case, the worst case
complexity for multiplicative weighted Voronoi diagrams is
$\Theta(n^2)$ \cite{ae-oacwv-84}, even in the plane. In the weighted
case, the cells are not necessarily connected, and a bisector of two
sites is either a line or an (Apollonius) circle.

In the \emph{Power diagram}, each site $\site_i$ has an associated
radius $r_i$, and the distance of a point $\pnt$ to this site is
$\distY{\site_i}{\pnt}^2 -r_i^2$; that is, the squared length of the
tangent from $\pnt$ to the disk of radius $r_i$ centered at $\site_i$.
As such, Power diagrams allow including weight in the distance
function, while still having bisectors that are straight lines and
having linear combinatorial complexity overall.

Klein \cite{k-avda-88} introduced (and this was further refined by
Klein \etal~\cite{kln-avdr-09}) the notion of \emph{abstract Voronoi
   diagrams} to help unify the ever growing list of variants of
Voronoi diagrams which have been considered.  Specifically, a simple
set of axioms was identified, focusing on the bisectors and the
regions they define, which classifies a large class of Voronoi
diagrams with linear complexity (hence such axioms are not intended to
model, for example, multiplicative diagrams).

\paragraph{Randomization and Expected Complexity.}

In many cases, there is a big discrepancy between the worst case
analysis of a structure (or an algorithm) and its average case
behavior. This suggests that in practice, the worst case is seldom
encountered.  For example, recently, Agarwal \etal~\cite{aks-urmsn-13,%
   ahks-urmsn-14}, showed that the expected union complexity of a set
of randomly expanded disjoint segments is $O(n\log n)$, while in the
worst case the union complexity can be quadratic. In other words,
Agarwal \etal bounded the expected complexity of a level set of the
randomly weighted Voronoi diagram of disjoint segments.

\paragraph{If the sites are placed randomly.}

There is extensive work on the expected complexity of various
structures (including Voronoi diagrams) if the sites are being picked
randomly (but not their weight), see \cite{s-iig-53,rs-udkhv-63,%
   r-slcdn-70,d-hdvdl-89,ww-sg-93,sw-ig-93,obsc-stcav-00,%
   \si{h-ecrch-11},\si{dhr-ecvdt-12}} (this list is in no way
exhaustive). In many of these cases the resulting expected complexity
is dramatically smaller than its worst case analysis. For example, in
$\Re^d$, the Voronoi diagram of $n$ sites picked uniformly inside a
hypercube has $O(n)$ complexity (the constant depends exponentially on
the dimension), but the worst case complexity is, as already
mentioned, $\Theta(n^{\ceil{d/2}})$.  Intuitively, the low complexity
when the locations are randomly sampled is the result of the relative
uniformity of such samples. Interestingly, there is a subtle
connection between such settings and the behavior of grid points
\cite{h-osafd-98a}.

However, in this paper, site locations will be fixed and site weights
will be sampled (similar to the model of Agarwal
\etal~\cite{ahks-urmsn-14}).  As such, our argument cannot rely on the
spacial uniformity provided by location sampling. %
Nevertheless, the case of fixed (distinct) weights and sampled
locations will follow readily from our arguments for the sampled
weights and fixed locations case, see \secref{locations}.

\paragraph{Technical Challenges in the Multiplicative Setting.}

In this paper, we focus on the case of multiplicative weighted Voronoi
diagrams in the plane.  As mentioned above, such diagrams can have
quadratic complexity. It is thus natural to consider sampling (of the
weights) as a way to mitigate this, and argue for lower expected
complexity.  However, the multiplicative case poses a significant
technical hurdle in that nearest weighted neighbor relations are a
non-local phenomena.  Specifically, for a given site, unless all its
neighbors (in the unweighted diagram) have lower weight, its region of
influence cannot be locally contained, and its cell spills over --
potentially affecting points far away.  This non-locality makes
arguing about such diagrams technically challenging.  For example, the
work of Agarwal \etal~\cite{ahks-urmsn-14} required quite a bit of
effort to bound the level set of the multiplicative Voronoi diagram
for segments, and it is unclear how their analysis can be extended to
bound the complexity of the whole diagram.

\subsection*{Our Results.}

Consider a fixed distribution from which we sample weights.  Our main
result is that the expected complexity of the multiplicative weighted
Voronoi diagram of a set of sites is near linear, where the sites are
disjoint simply-connected compact regions in the plane.  We specify
the exact requirements on the sites in \secref{prelims} --- possible
sites include points, segments, or convex sets.

A simple consequence of our main result is that the expected
complexity of the union of randomly expanded disjoint segments or
convex sets is also near linear.  Specifically, our proof is
significantly simpler than the one of Agarwal
\etal~\cite{ahks-urmsn-14}.  Our bound is weaker by (roughly) an
$O(\log n)$ factor for the case of segments, but for convex sets we
improve the bound from $O(n^{1+\eps})$ to
$O(n \operatorname{polylog} n)$ (and our bound holds for the
complexity of the whole diagram, not only the level set).  Also,
similar to the work of Agarwal \etal, in \secref{permutation} we make
the observation that our results also hold for the more general case
where instead of sampling weights from a distribution, one is given a
fixed set of $n$ weights which are randomly permuted among the sites.

Our technique is rather versatile and should be applicable to other
well behaved distance metrics (for example, when each site has its own
additive constant which is included when measuring the distance to
that site).

To extend our result to more general sites, we prove that in these
settings the expected complexity of the overlay of the Voronoi cells
in a randomized incremental construction is $
O(\dsFunc{\intersections}{n} \log n)$ (see \lemrefpage{overlay:2}),
where $\dsFunc{\intersections}{n}$ is the length of a
Davenport-Schinzel sequence of order $\intersections$ with $n$
symbols, where $\intersections$ is some constant. This is an extension
of the result of Kaplan \etal~\cite{krs-omdri-11} to these more
general settings.

\paragraph{Significance of Results.}

As discussed above, due to non-locality, analyzing multiplicative
diagrams seems challenging.  In particular, we are unaware of any
previous subquadratic bounds for the expected complexity.  On the
practical side, the unwieldy complexity of the multiplicative diagram
(and its lack of a dual structure, similar to Delaunay triangulations)
has discouraged their use in favor of more well behaved diagrams, such
as the power diagram.  Our work indicates that using such diagrams in
the real world might be practical, despite their worst case quadratic
complexity.  In particular, our technique for bounding the expected
complexity immediately implies a near linear time randomized
incremental construction algorithm for computing the multiplicative
diagram.

\paragraph{Outline of technique.}

Consider the case of bounding the expected complexity of the Voronoi
diagram of a set $\PntSet$ of $n$ multiplicatively weighted points
(i.e., sites) in the plane, where the weights are being picked
independently from the same distribution. The key ideas behind the new
approach are the following.
\begin{compactenum}[\;\;(A)]
    \item \textbf{Candidate Sets.} Consider any point $\query$ in the
    plane, and let $\pnt$ be its nearest neighbor in $\PntSet$ under
    the weighted distance.  Now, if $\pnt$ is the nearest neighbor of
    $\query$ then for all other sites in $\PntSet$ either $\pnt$ has
    smaller weight, or smaller distance to $\query$.  Thus for each
    point $\query$ in the plane one can define its candidate set,
    which consists of all sites $\pntB \in \PntSet$ such that for all
    other sites in $\PntSet$ either $\pntB$ has smaller weight or
    smaller distance to $\query$.  Saying it somewhat differently,
    plotting the points of $\PntSet$ in the parametric plane, where
    one axis is the distance from $\query$, and the other axis is
    their weight, the candidate set is all the minima points (i.e.,
    they are the vertices of the lower left staircase of the point
    set, and they are not dominated in both axes by any other point).
    We show that when weights are randomly sampled, with high
    probability, for all points in the plane the candidate set has at
    most logarithmic size (this is well known, and we include the
    proof for the sake of completeness).
    
    \item \textbf{Gerrymandering the plane.}  Next, we partition the
    plane into a small number of regions such that the candidate set
    is fixed within each region.  Specifically, if one can break the
    plane into $m$ such uniform candidate regions, then the worst case
    complexity of the Voronoi diagram is $O(m\log^2 n)$, with high
    probability, since all candidate sets are of size at most $O( \log
    n )$, and the worst case complexity of the multiplicative Voronoi
    diagram of a weighted set of points is quadratic.
    
    \item \textbf{Randomized Incremental Campaigning.} %
    The main challenge, as frequently is the case, is to do the
    gerrymandering.  To this end, consider adding the sites in order
    of increasing weight.  When the $i$\th site is added, it has
    higher weight than the sites already added, and lower weight than
    the sites which have not been added yet.
       Therefore, the $i$\th site is in the candidate set of a point in
       the plane, when it is the nearest neighbor of the point among the
       first $i$ sites.
    In other words, the points in the Voronoi cell of the $i$\th site
    in the Voronoi diagram of the first $i$ sites.  Next, consider the
    overlay of the $n$ Voronoi cells formed by taking all
    $i \in \brc{1, \ldots, n}$. Observe that each face of this overlay
    has the same candidate set.  For the case of points, Kaplan
    \etal~\cite{krs-omdri-11} proved that this overlay has
    $O(n \log n)$ expected complexity. This implies immediately an
    $O(n \log^3 n)$ bound on the expected complexity of the
    multiplicative Voronoi diagram.
\end{compactenum}

\paragraph{Organization.}

In \secref{prelims} we introduce notation and definitions used
throughout the paper.  In \secref{result} we introduce the notion of
candidate sets, and show how partitioning the plane into a near linear
number of regions, such that each region has the same candidate set
implies our result on the near linear expected complexity of the
multiplicative Voronoi diagram of sites.  Specifically, the
partitioning used is the overlay of Voronoi cells in a randomized
incremental construction (\RIC), and in \secref{proof} we describe in
detail how the expected complexity of such an overlay is near linear.
In \secref{resultApplications} we state our main result, and present a
number of specific applications of our technique.  In
\secref{permutation} we observe that instead of sampling weights our
technique extends to the more general case of permuting a fixed set of
weights among the points.  In \apndref{lower:bound}, we show that the
overlay of Voronoi cells in \RIC is $\Omega(n \log n)$, implying that
the upper bound of Kaplan \etal~\cite{krs-omdri-11} is tight in this
case.

\section{Preliminaries}
\seclab{prelims}

Below we define Voronoi diagrams and related objects in a rather
general way to encompass the various applications of our technique.
For simplicity the reader is encouraged to interpret these definitions
in terms of Voronoi diagrams of points (or less trivially disjoint
segments).

\paragraph{Notation.}

We use $\OSiteSet = \permut{\site_1, \ldots, \site_n}$ to denote a
permutation of a set $\SiteSet$ of $n$ objects, and
$\OSiteSet_i = \permut{\site_1, \ldots, \site_i}$ to denote the
\emphi{prefix} of this permutation of length $i$. Similarly, we use
$\OSuffix{i+1} = \permut{\site_{i+1}, \site_{i+2}, \ldots, \site_n} $
to denote a \emphi{suffix} of $\OSiteSet$. When we care only about
what elements appear in a permutation, $\OSiteSet$, but not their
internal ordering, we use the notation $\SiteSet = \SetOf{ \OSiteSet}$
to denote the associated set.  As such,
$\SiteSet_i = \SetOf{\OSiteSet_i}$ is the \emphi{unordered prefix} of
length $i$ of $\OSiteSet$, and
$\Suffix{i+1} = \SetOf{\OSuffix{i+1}} = \brc{\site_{i+1}, \site_{i+2},
   \ldots, \site_n} $ is the \emphi{unordered suffix}.
   
\paragraph{Arrangements.}

As it will be used throughout the paper, we now define the standard
terminology of arrangements (see \cite{sa-dsstg-95,h-gaa-11}).  Given
a set $\mathsf{S}$ of $n$ segments in the plane, its
\emphi{arrangement}, denoted by $\Arr( \mathsf{S} )$, is the
decomposition of the plane into faces, edges and vertices. The
vertices $\Arr( \mathsf{S} )$ are the endpoints and the intersection
points of the segments of $\mathsf{S}$, the edges are the maximal
connected portions of the segments not containing any vertex, and the
faces are the connected components of the complement of the union of
the segments of $\mathsf{S}$.  For a set of polygons, we can
analogously define its arrangement by letting $\mathsf{S}$ be the
union of all boundary segments of the polygons.

\subsection{Voronoi diagrams}

Let $\SiteSet = \brc{\site_1, \ldots, \site_n}$ be a set of $n$ sites
in the plane.  Specifically, the sites are disjoint simply-connected
compact subsets of $\Re^2$.  For a closed set $Y \subseteq \Re^2$, and
any point $\query \in \Re^2$, let
$\distSet{\query}{Y} = \min_{y \in Y} \dist{\query}{y}$ denote the
\emphi{distance} of $\query$ to the set $Y$.  For any two sites
$\site, \siteA\in \SiteSet$, we define their \emphi{bisector}
$\bisector(\site, \siteA)$ as the set of points $\query\in \Re^2$ such
that $\distSet{\query}{\site} = \distSet{\query}{\siteA}$. Each
$\site\in \SiteSet$, induces the function
\begin{math}
    \func_\site(\query) = \distSet{\query}{\site},
\end{math}
where $\query$ is any point in the plane.  For any subset
$\sub\subseteq \SiteSet$ and any site $\site\in \sub$,
the \emphi{Voronoi cell} of $\site$ with respect to $\sub$,
denoted $\VorCell{\site}{\sub}$, is the subset of $\Re^2$ whose closest site
in $\sub$ is $\site$, i.e.
\begin{math}
    \VorCell{\site}{\sub} = \Set{ \query \in \Re^2}{ \forall \siteA\in
       \sub \;\; \func_\site(\query) \leq \func_\siteA(\query)}%
    \Bigr.%
    .
\end{math}
Finally, for any subset $\sub\subseteq \SiteSet$, 
the Voronoi diagram of $\sub$, denoted $\VorX{\sub}$, is the
partition of the plane into Voronoi cells induced by the minimization
diagram of the set of functions
$\Set{\func_\site}{\bigl. \site\in \sub}$.


\begin{remark}%
    \remlab{assumptions}%
    Throughout the paper we require the following from the bisectors
    and Voronoi cells.
    \begin{compactenum}[\qquad (A)]
        \item For any two sites $\site, \siteA\in \SiteSet$, their
        bisector $\bisector(\site, \siteA)$ is a simple curve
        (i.e. the image of a continuous map from the unit interval to
        $\Re^2$) whose removal splits the plane into exactly two
        unbounded regions\footnote{That is, under the stereographic
           projection of the plane to the sphere, the bisector is a
           simple closed Jordan curve through the north pole.}.
        \item Each bisector is of constant complexity, that is it has
        a constant number of extremal points in the direction of (say)
        the $x$-axis\footnote{One can assume the bisectors contain no
           vertical segments, since otherwise we can slightly rotate
           the plane, and hence extremal points are well defined.}.
        \item Any two distinct bisectors intersect at most a constant
        number of times.
        \item For any site $\site\in \SiteSet$ and any subset
        $\sub\subseteq \SiteSet$, the set $\VorCell{\site}{\sub}$ is a
        simply connected subset of the plane.
        \item For any subset $\sub\subseteq \SiteSet$,
           the Voronoi cells cover the plane; that is,
           $\cup_{\site \in \sub} \,\VorCell{\site}{\sub} = \Re^2$.
    \end{compactenum}
\end{remark}

One can view the union, $U$, of the boundaries of the cells in a
Voronoi diagram as a planar graph.  Specifically, define a
\emphi{Voronoi vertex} as any point in $U$ which is equidistant to
three sites in $\SiteSet$ (which happens at the intersection of two
bisectors).  For simplicity, we make the general position assumption
that no point is equidistant to four or more sites in the plane.
Furthermore, define a \emphi{Voronoi edge} as any maximal connected
subset of $U$ which does not contain a Voronoi vertex.  (Note that in
order for each edge to have two endpoints we must include the
``point'' at infinity, i.e. the graph is defined on the stereographic
projection of the plane onto the sphere.)

The above conditions imply that the Voronoi diagram of any subset of
sites is an \emphi{abstract Voronoi diagram} (actually such diagrams
are more general).  It is known that for such diagrams
the overall complexity of the Voronoi graph they define is linear
\cite{kln-avdr-09}, i.e. for a set of $n$ sites the number of Voronoi
vertices, edges, and faces is $O(n)$.  In general Voronoi diagrams,
edges may have more than a constant number of $x$-extremal points.
However, since we assumed each bisector has a constant number of
$x-$extremal points, and Voronoi edges are contiguous subsets of
bisectors, there are only a constant number of $x$-extremal points on
any edge.  Therefore, asymptotically, the complexity of
   the vertical decomposition of the Voronoi diagram, and thus the
   diagram itself, is bounded by the complexity of this (Voronoi)
   graph.

\subsection{Multiplicative weighted Voronoi diagrams}

As before, let $\SiteSet$ be a set of $n$ weighted sites in the plane,
where $\weight_i > 0$ is the weight associated with the $i$\th site
$\site_i$.  Consider the weighted Voronoi diagram of $\SiteSet$,
denoted by $\WVorX{\SiteSet}$. Specifically, for $i=1,\ldots, n$, the
site $\site_i$ induces a distance function $\func_i(\query) =
\weight_i \, \distSet{\query}{\site}$. The \emphi{multiplicative
   weighted Voronoi} diagram of $\SiteSet$ is the partition of the
plane induced by the minimization diagram of the distance functions
$\func_1, \ldots, \func_n$. The \emphi{weighted Voronoi cell} of
$\site_i$ is
\begin{align}
    \VCell_i = \Set{ \query \in \Re^2 }{ \forall j \quad
       \func_i(\query) \leq \func_j(\query)}.
    \eqlab{w:cell}
\end{align}
For a multiplicative Voronoi diagram, the cells are not necessarily
connected.

\begin{remark}%
    \remlab{assumptions:2}%
    In addition to the conditions listed in \remref{assumptions}, for
    the unweighted case, we also require the following of the weighted
    diagram (for any positive weight assignment):
    \begin{compactenum}[\qquad(A)]
        \item Each weighted bisector has a constant number of extremal
        points in the direction of the $x$-axis.
        \item Any two distinct weighted bisectors intersect at most a
        constant number of times.
    \end{compactenum}
\end{remark}


Let $\worstX{n}$ denote the worst case complexity of the
multiplicative weighted Voronoi diagram.  The analysis below requires
a polynomial bound on $\worstX{n}$.  It is not hard to see that the
conditions above on bisectors already imply a bound of $\worstX{n} =
O(n^4)$.

\subsubsection{Assigning weights randomly}

In the following, we use \emphi{distribution} to refer to any
probability distribution defined over the positive real numbers (i.e.,
$\Re^+$). We use $\pdf$ to denote this distribution, which might be
continuous or discrete.

Let $\SiteSet$ be a given set of $n$ sites in the plane.  We assign
each site of $\SiteSet$ a random weight sampled independently from
$\pdf$.  We order the sites of $\SiteSet$ by their weight, and let
$\site_i$ be the site assigned the $i$\th smallest weight, and let
$\weight_i$ denote this weight, for $i=1, \ldots, n$. If the weights
assigned are not unique, we randomly permute each cluster of sites
that are assigned the same weight internally\footnote{Specifically,
   for every site $\site$ generate, in addition to its weight
   $\weight$ chosen from $\pdf$, a secondary weight $\weight'$ which
   is picked uniformly at random from the interval $[0,1]$. Now order
   the sites in lexicographical ordering of the pairs $(\weight,
   \weight')$.}.  The resulting ordering $\OSiteSet = \permut{\site_1,
   \ldots,\site_n}$ is a (uniform) random permutation defined over the
sites of $\SiteSet$.

\section{Bounding the complexity of the randomly %
   weighted diagram}
\seclab{result}

Let $\SiteSet$ be a weighted set of sites in the plane, whose ordering
by increasing weight is $\OSiteSet = \permut{\site_1, \ldots,
   \site_n}$ (where $\weight_i$ is the weight of $\site_i$).  For any
point $\query \in \Re^2$, we write $\cellPL{\query}{\SiteSet}$ to
denote the Voronoi cell of $\WVorX{\SiteSet}$ that contains $\query$,
i.e.  $\cellPL{\query}{\SiteSet} = \VCell_i$ is induced by the site
$\site_i = \arg \min_{\site_j \in \SiteSet}
\weight_j\dist{\query}{\site_j}$, see \Eqref{w:cell} (if $\query$ is a
boundary point, then we arbitrarily pick one of the equidistant
sites).

\subsection{\Flirting sets}

\begin{defn}
    Let $\OSiteSet = \permut{\site_1, \ldots, \site_n}$ be an ordered
    set of $n$ sites in the plane. For any point $\query$ in the
    plane, the \emphi{\flirting set} of $\query$, denoted by
    $\candidY{\query}{\OSiteSet}$, is the set of all sites $\site_i
    \in \OSiteSet$, such that $\dist{\query}{\site_i} =
    \distSet{\query}{\OSiteSet_i}$, for $i=1,\ldots, n$.  In words,
    $\site_i$ is in $\candidY{\query}{\OSiteSet}$ if it is the closest
    site to $\query$ in its prefix $\OSiteSet_i$.
\end{defn}

A prerequisite for a site $\site_j$ of the weighted site set
$\SiteSet$ to be the nearest site under weighted distances to
$\query$, is that $\site_j$ is in the \flirting set
$\candidY{\query}{\OSiteSet}$.

\begin{lemma}
    \lemlab{influence}%
    %
    For a point $\query$ in the plane, if
    $\,\cellPL{\query}{\SiteSet} = \cell_j$, then $\site_j$ is in
    $\candidY{\query}{\OSiteSet}$, where $\OSiteSet$ is the ordering
    of $\SiteSet$ by increasing weight.
\end{lemma}
\begin{proof}
    Let $\site_j$ be the nearest weighted site to $\query$ (in the
    weighted Voronoi diagram of $\SiteSet$).  Consider any other site
    $\site_i$ such that $\weight_i<\weight_j$. This implies that $i<j$
    because of the ordering of $\OSiteSet$.  Observe that $\site_i$ is
    further away from $\query$ than $\site_j$ (i.e.,
    $\dist{\query}{\site_j} < \dist{\query}{\site_i}$), since
    otherwise
    \begin{align*}
        \func_j\pth{\query} = \weight_j \dist{\query}{\site_j} \geq
        \weight_j \dist{\query}{\site_i} > \weight_i
        \dist{\query}{\site_i} = \func_i\pth{\query},
    \end{align*}
    which is a contradiction.  In other words, $\site_j$ must be the
    (unweighted) closest point to $\query$ in its prefix
    $\OSiteSet_j$.  %
    \DCGVer{{\qed}}
\end{proof}

We next prove that, with high probability, the candidate set is
logarithmic in size for all points in the plane. To this end, we need
the following helper lemma.

\newcommand{\SuffixEvent}{F}
\newcommand{\SuffixE}[1]{\SuffixEvent_{#1}}
\newcommand{\SuffixSet}{\mathcal{F}}

\begin{lemma}
    \lemlab{high:prob}%
    Let $\Pi=\permut{\pi_1, \ldots, \pi_n}$ be a random permutation of
    $\brc{1, \ldots, n}$, and let $X_i$ be an indicator variable which
    is $1$ if $\pi_i$ is the smallest number among
    $\pi_1, \ldots, \pi_i$, for $i=1,\ldots, n$.  Let
    $Z=\sum_{i=1}^n X_i$, then $Z=O(\log n)$, with high probability
    (i.e., $\geq 1- 1/n^{c}$, for any constant $c$).
\end{lemma}

\begin{proof}
    This is well known \cite[Section 3.4]{m-cgitr-94}, and we include
    the proof for the sake of completeness.  Let $\Event_i$ be the
    event that $X_i=1$, for $i=1,\ldots, n$.  We first show the events
    $\Event_1, \ldots, \Event_n$ are independent (implying the $X_i$
    are independent). Indeed, conceptually, generate the permutation
    as follows: Randomly pick a permutation of the given numbers, and
    set the first number to be $\pi_n$. Next, pick a random
    permutation of the remaining numbers and set the first number as
    the penultimate number (i.e., $\pi_{n-1}$) in the output
    permutation. Repeat this process till we generate the whole
    permutation.

    Observe that by our thought experiment, regardless of the elements
    that appear in the suffix $\permut{\pi_{i_1+1},\ldots,\pi_n}$,
    there is exactly one minimum value in the remaining elements, and
    these remaining elements are randomly permuted before determining
    $\pi_i$.  Now, consider arbitrary indices
    $1 \leq i_1 < i_2 < \ldots < i_k \leq n$.  Clearly, the event
    $\Event_{i+1}$ is not influenced by the exact choice of the suffix
    $\pi_{i+1}, \ldots, \pi_{n}$, and as such
    \begin{math}
        \Prob{ \Event_{i_1} \!\sep{ \pi_{i_1+1}, \ldots, \pi_n}}%
        =%
        \Prob{\bigl. \Event_{i_1}}%
        =%
        1/i_1.
    \end{math}
    Furthermore, we have
    \begin{math}
        \Prob{ \Event_{i_1} \sep{ \Event_{i_2} \cap \ldots \cap
              \Event_{i_k} }}%
        =%
        \Prob{\bigl. \Event_{i_1}}%
        =%
        1/i_1
    \end{math}
    as can be easily verified\footnote{A formal proof of this is
       somewhat tedious. Indeed, given $y_{t}, \ldots, y_n$, let
       \begin{math}
           \Bigl.  \SuffixE{t} = \SuffixE{t}\pth{ y_{t}, \ldots, y_n}
           = \bigl\{ \permut{\pi_1,\ldots, \pi_n} \bigm| \pi_{t} =
           y_{t}, \pi_{t+1} = y_{t+1}, \ldots, \pi_{n} = y_{n} \bigr\}
       \end{math}
       denote the \emph{suffix event}, where the specific values of
       the $\pi_{t},\ldots, \pi_n$ are fully specified.  By the
       above, we have
       \begin{math}
           \Prob{ \bigl.  \Event_{i_1} \sep{ \SuffixE{i_1+1} }} =
           1/i_1.
       \end{math}
       Observe that the event
       \begin{math}
           \Event_{i_2} \cap \ldots \cap \Event_{i_k}
       \end{math}
       is the disjoint union of suffix events. Indeed, for a specific
       value of $y_{i_2},\ldots, y_n$, either all the permutations of
       $\SuffixE{i_2}(y_{i_2},\ldots, y_n)$ or none, are in
       $\Event_{i_2} \cap \ldots \cap \Event_{i_k}$.  As such, let
       $\SuffixSet$ be the set of all the suffix events
       that are in
       $\Event_{i_2} \cap \ldots \cap \Event_{i_k}$.  Now, we have
       \begin{math}
           \Prob{ \Event_{i_1} \sep{ \Event_{i_2} \cap \ldots \cap
                 \Event_{i_k} }}%
           =%
           \sum_{\SuffixEvent \in \SuffixSet} \Prob{
              \Event_{i_1} \sep{ \SuffixEvent}} \Prob{\SuffixEvent
              \sep{ \Event_{i_2} \cap \ldots \cap \Event_{i_k}}}%
           =%
           \sum_{\SuffixEvent \in \SuffixSet} \Prob{
              \Event_{i_1} \sep{ \SuffixEvent}} \Prob{\SuffixEvent
              \sep{ \Event_{i_2} \cap \ldots \cap \Event_{i_k}}}%
           =%
           {1}/{i_i}.
       \end{math}
       (This is similar in spirit to arguments used in martingales,
       where $F_n, F_{n-1}, \ldots$ is a filter.)  }.
    %
    As such, by induction, we have
    \begin{align*}
        \Prob{\bigl. \bigcap\nolimits_{j=1}^k \Event_{i_j} }%
        &=%
        \Prob{\bigl. \Event_{i_1} \sep{ \bigcap\nolimits_{j=2}^k
              \Event_{i_j} }}%
        \Prob{ \bigcap\nolimits_{j=2}^k \Event_{i_j}}%
        %
        \DCGVer{\\&}%
        =%
        \Prob{\bigl. \Event_{i_1}}%
        \Bigl. \ProbChar \bigl[ \bigcap\nolimits_{j=2}^k \Event_{i_j}
        \bigr] \biggr.%
        =%
        \prod_{j=1}^k \Prob{\Event_{i_j}}%
        =%
        \prod_{j=1}^k \frac{1}{i_j}.
    \end{align*}
    We conclude that the variables $X_1, \ldots, X_n$ are independent.
    The claim now follows from the Chernoff bound since
    \begin{math}
        \mu%
        =%
        \Ex{\bigl. Z}%
        =%
        \sum_i \Bigl. \Ex{\bigl. X_i}%
        =%
        \sum_{i=1}^n {1}/{i}%
        =%
        \Theta(\log n).
    \end{math}%
    \DCGVer{{\qed}}
\end{proof}

\begin{corollary}
    \corlab{forAll}%
    Let $\SiteSet$ be a randomly weighted set of $n$ sites in the
    plane, and let $\OSiteSet = \permut{\site_1, \ldots, \site_n}$ be
    the sorted ordering of $\SiteSet$ by increasing weight.
    Simultaneously for all points in the plane their \flirting set for
    $\OSiteSet$ is of size $O(\log n)$, with high probability.
\end{corollary}
\begin{proof}
    Consider any fixed point $\query$ in the plane.  Since
    $\OSiteSet = \permut{\site_1, \ldots, \site_n}$ is a random
    permutation of $\SiteSet$, the sequence
    $\dist{\query}{\site_1}, \ldots, \allowbreak
    \dist{\query}{\site_n}$
    is a random permutation of the distance values from $\query$ to
    the sites in $\SiteSet$.  Therefore, by the definition of
    the \flirting set and \lemref{high:prob}, we have
    $\cardin{\candidY{\query}{\OSiteSet}} = O(\log n)$ with high
    probability.
    
    Consider the arrangement of all the (unweighted) bisectors of all
    the pairs of sites in $\SiteSet$.  There are $n$ sites and
    $\binom{n}{2}$ bisectors.  As such, there are $O\pth{n^4}$
    vertices in this arrangement, as by assumption each pair of
    bisectors intersect at most a constant number of times, and each
    bisector has a constant number of $x$-extremal points.
       Therefore, the total complexity of this arrangement is
       $O\pth{n^4}$.

    Within each face of this arrangement, the \flirting set cannot
    change since all points in this face have the same ordering of
    their distances to the sites in $\SiteSet$.  So pick a
    representative point for each of the $O(n^4)$ faces.  For any such
    representative, with probability $\leq 1/n^c$, the \flirting
    set has $>\alpha\log(n)$ sites, for any constant $c$ of our
    choosing (where $\alpha$ is a constant determined by the Chernoff
    bound that depends only on $c$).  Therefore, by choosing $c$ to
       be sufficiently large, taking the union bound on
    these bad events, and then taking the complement, the claim
    follows. %
    \DCGVer{{\qed}}
\end{proof}

\subsection{Getting a compatible partition}
\seclab{compat}

The goal now is to find a low complexity subdivision of the plane,
such that within each cell of the subdivision the \flirting set is
fixed.  The main insight is that by using the unweighted Voronoi
diagram one can get such a subdivision.

Let $\VCellPrefix{i}$ denote the Voronoi cell of $\site_i$ in the
unweighted Voronoi diagram of the $i$\th prefix
$\SiteSet_i = \brc{\site_1,\ldots, \site_i}$. Let $\Arr$ denote
the arrangement formed by the overlay of the regions
$\VCellPrefix{1}, \ldots, \VCellPrefix{n}$.  The complexity of $\Arr$,
denoted by $\cardin{\Arr}$, is the total number of these faces, edges,
and vertices, as well as the number of $x$-extremal points on the
edges.  By our assumptions on the bisectors, the number of vertices
bounds the complexity $\cardin{\Arr}$.

\begin{lemma}
    \lemlab{constant}%
    For any face $\face$ of $\Arr = \Arr\pth{ \VCellPrefix{1}, \ldots,
       \VCellPrefix{n}}$, the \flirting set is the same, for all
    points in $\face$.
\end{lemma}

\begin{proof}
    Initially, all points in the plane have the same \flirting set,
    namely the empty set.  When the site $\site_i$ is added, the only
    points in the plane whose \flirting set changes are those such
    that $\site_i$ is their nearest neighbor in $\SiteSet_i$.
    However, these are precisely the points in the Voronoi cell of
    $\site_i$ in the unweighted Voronoi diagram of $\SiteSet_i$. That
    is, the \flirting set changes only for the points covered by
    $\VCellPrefix{i}$ -- where $\site_i$ is being added to the
    \flirting set.
    
    The claim now easily follows, as $\Arr$ is the overlay arrangement
    of these regions. %
    \DCGVer{{\qed}}
\end{proof}


\begin{theorem}
    \thmlab{sitesTheorem}%
    %
    Let $\SiteSet$ be a set of $n$ sites in the plane, satisfying the
    conditions in \remref{assumptions} and \remref{assumptions:2},
    where for each site a weight is sampled independently from some
    distribution $\pdf$.  Let $\OSiteSet = \permut{\site_1, \ldots,
       \site_n}$ be the ordering of the sites by increasing weights,
    and let $\VCellPrefix{i} = \VorCell{\site_i}{\OSiteSet_i}$, for
    $i=1,\ldots, n$. Let $\Arr = \Arr\pth{ \VCellPrefix{1}, \ldots,
       \VCellPrefix{n}}$ be the arrangement formed by the overlay of
    all these cells.
    
    Then, the expected complexity of the multiplicative Voronoi
    diagram $\WVorX{\SiteSet}$ is
    $O\pth{\Bigl. \Ex{ \bigl.  \cardin{\Arr} }\worstX{\log n} }$,
    where $\cardin{\Arr}$ is the total complexity of $\Arr$, and
    $\worstX{m}$ denotes the worst case complexity of a weighted
    Voronoi diagram of $m$ sites.
\end{theorem}
\begin{proof}
    We first compute a vertical decomposition of the faces of $\Arr$, 
    in order to break up each face into constant complexity cells.
    Specifically, each face has two types of vertices -- extremal
    points on the bisectors and intersections of bisectors. From each
    such vertex shoot out vertical rays.  Doing so partitions the
    plane into constant complexity cells (or, somewhat imprecisely,
    vertical trapezoids) and the total number of such cells is
    proportional to $\cardin{\Arr}$ (i.e. the number of extremal
    points and intersections).
    
    \lemref{constant} implies that within each cell of the vertical
    decomposition the \flirting set is fixed.  So consider such a cell
    $\Delta$, and let $\candidSet$ be its \flirting set.
    \lemref{influence} implies that the only sites whose weighted
    Voronoi cells can have non-zero area in $\Delta$ are the sites in
    $\candidSet$.  That is, the Voronoi diagram in $\Delta$ is the
    intersection of $\Delta$ with the weighted Voronoi diagram of some
    subset of $\candidSet$.  Now the weighted Voronoi diagram of
    $\leq \cardin{\candidSet}$ points has worst case complexity
    $\worstX{\cardin{\candidSet}}$.  Since $\Delta$ is a constant
    complexity region this implies that the complexity of the weighted
    Voronoi diagram in $\Delta$ is
    $O\pth{ \worstX{ \cardin{ \candidSet } } }$.
    
    By \corref{forAll}, for all points in the plane, the \flirting set
    is of size $O( \log n)$ (with high probability), and since there
    are $O( \cardin{\Arr})$ cells (in expectation), the claim now
    readily follows. %
    \DCGVer{{\qed}}
\end{proof}

\newcommand{\VVX}[1]{%
   \begin{minipage}{0.297\linewidth}%
       \begin{minipage}{0.98\linewidth}%
           \includegraphics[page=#1,width=0.99\linewidth]{figs/overlay_faces}%
           \hspace*{-1cm}
      \end{minipage}
   \end{minipage}
}%
\newcommand{\VVZ}[1]{%
   \includegraphics[page=#1,width=0.5\linewidth]{figs/overlay_faces}%
}

\begin{figure}[\DCGVer{t}\NotDCGVer{p}]
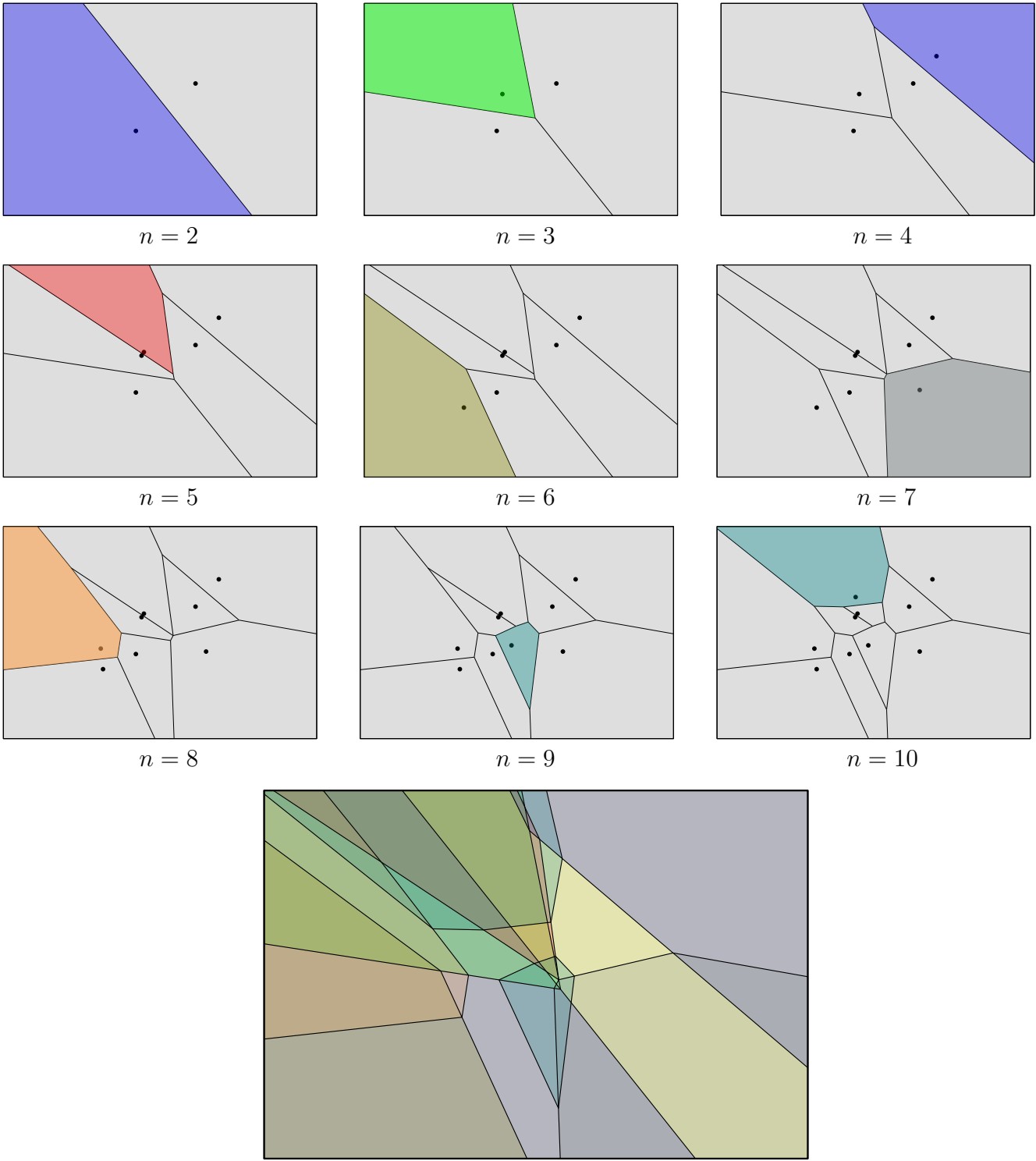

    \begin{tabular}{c%
      cc}
      \VVX{2} & \VVX{3}& \VVX{4}%
                         \smallskip
      \\
      $n=2$ & $n=3$ & $n=4$
                      \medskip%
      \\
      \VVX{5} & \VVX{6}& \VVX{7} 
                         \smallskip
      \\
      $n=5$ & $n=6$ & $n=7$%
                      \medskip%
      \\
      \VVX{8} & \VVX{9} & \VVX{10}
                          \smallskip%
      \\
      $n=8$ & $n=9$ & $n=10$ 
                      \medskip
    \end{tabular}
    \centerline{\VVZ{12}}
    \caption{The randomized incremental construction of a Voronoi
       diagram of point sites, and the resulting overlay arrangement.}
    \figlab{overlay:Voronoi}
\end{figure}

For the concrete case when the sites are points in the plane, the
worst case complexity of the weighted Voronoi diagram is quadratic
\cite{ae-oacwv-84}, and so in the above theorem $\worstX{m} = O(m^2)$.
Kaplan \etal \cite{krs-omdri-11} showed that for a random permutation
of $n$ points (as is the case here) the expected total complexity of
$\Arr$ is $O(n\log n)$, see \figref{overlay:Voronoi} for an example of
such an overlay arrangement.  We therefore readily have the following
result.

\begin{theorem}
    \thmlab{points}%
    Let $\PntSet$ be a set of $n$ points in the plane, where for each
    point we independently sample a weight from some distribution
    $\pdf$.  Then the expected complexity of the multiplicative
    Voronoi diagram of $\PntSet$ is $O\pth{ n\log^3 n }$.
\end{theorem}

\begin{corollary}
    Let $\PntSet$ be a set of $n$ points in the plane, where for each
    point we independently sample a weight from some distribution
    $\pdf$.  Then the multiplicative Voronoi diagram of $\PntSet$ can
    be computed in $O\pth{ n\log^3 n }$ expected time.
\end{corollary}
\begin{proof}
    This follows readily from the above constructive proof, and so we
    only sketch the algorithm. First, compute the ordering
    $\OSiteSet =\permut{\site_1, \ldots,\site_n}$ of the sites by
    increasing weight, and the set of polygons
    $\brc{\VCellPrefix{1}, \ldots, \VCellPrefix{n}}$, where
    $\VCellPrefix{i} = \VorCell{\site_i}{\OSiteSet_i}$, by computing
    the unweighted Voronoi diagram by incremental construction
    (i.e. each $\VCellPrefix{i}$ is computed explicitly during the
    insertion of $\site_i$).  Next, compute
    $\Arr = \Arr\pth{ \VCellPrefix{1}, \ldots, \VCellPrefix{n}}$.
    Triangulate the faces of $\Arr$, and within each triangle compute
    the multiplicative Voronoi diagram of its candidate list, and clip
    it to the triangle (note the candidate lists are given by $\Arr$).

    For the running time, computing the unweighted Voronoi diagram by
    randomized incremental construction takes $O(n\log n)$ expected
    time.  By Kaplan \etal~\cite{krs-omdri-11}, the total number of
    segments over all the polygons and the arrangement $\Arr$ have
    expected complexity $O(n\log n)$.  Therefore computing $\Arr$ from
    $\brc{\VCellPrefix{1}, \ldots, \VCellPrefix{n}}$ takes
    $O(n\log^2 n)$ expected time \cite{m-cgitr-94}.  Triangulating the
    faces takes linear time in $\cardin{\Arr}$.  Using the quadratic
    time algorithm of Aurenhammer \cite{a-pdpaa-87}, computing the
    multiplicative diagram in each face take $O(\log^2 n)$ time, as by
    \corref{forAll}, all candidate lists have size $O(\log n)$. %
    \DCGVer{{\qed}}
\end{proof}

For more general sites, the real difficulty is in bounding
$\Ex{\cardin{\Arr}}$.  Specifically, in the next section we extend the
result of Kaplan \etal~\cite{krs-omdri-11} to these more general
settings.

\section{The complexity of the overlay of Voronoi %
   cells in \RIC}
\seclab{proof}

We next study the expected complexity of the overlay of Voronoi cells
and envelopes in a randomized incremental construction.  Specifically,
we first prove a result on the lower envelope of functions in two
dimensions, and then use it to prove a bound on the complexity of the
overlay of Voronoi cells of sites in the plane.

\subsection{Preliminaries}
\seclab{clarkson}

In the following, we need to use the Clarkson-Shor technique
\cite{cs-arscg-89}, which we quickly review here (see \cite{h-gaa-11}
for details).  Specifically, let $\ElemSet$ be a set of elements such
that any subset $\RSample \subseteq \ElemSet$ defines a corresponding
set of objects $\ObjX{\RSample}$ (e.g., $\ElemSet$ is a set of points
or sites in the plane, and any subset $R \subseteq \ElemSet$ induces
the set of edges of the Voronoi diagram $\VorX{\RSample}$).  Each
potential object, $\object$, has a defining set and a stopping set.
The \emphi{defining set}, $\DefSet{\object}$, is a subset of
$\ElemSet$ that must appear in $\RSample$ in order for the object to
be present in $\ObjX{\RSample}$, where this set has size bounded by
the same constant for all objects.  The \emphi{stopping set},
$\KillSet{\object}$, is a subset of $\ElemSet$ such that if any of its
members appear in $\RSample$ then $\object$ is not present in
$\ObjX{\RSample}$ (we also naturally require that $\KillSet{\object}
\cap \DefSet{\object} = \emptyset$, for all $\object$).  Surprisingly,
this already implies the following.

\begin{theorem}[Bounded Moments, \cite{cs-arscg-89}]%
    \thmlab{moments}%
    Using the above notation, let $\ElemSet$ be a set of $n$ elements,
    and let $\RSample$ be a random sample of size $r$ from
    $\ElemSet$. Let $f(\cdot)$ be a polynomially growing
    function\footnote{%
       A function $f(n)$ is a \emphi{polynomially growing}, if (i)
       $f(\cdot)$ is monotonically increasing, (ii) for any integers
       $i,n \geq 1$, $f(in) = i^{O(1)}f(n)$. This holds for example if
       $f(n)$ is a constant degree polynomial of $n$, with all its
       coefficients being positive. Of course, it holds for a much
       larger family of functions, e.g. $f(i) = i \log i$.}.  We have
    that 
    \begin{math}
        \ds%
        \Ex{\sum\nolimits_{\object\in \ObjX{\RSample}} f\pth{\bigl.
              \cardin{\KillSet{\object}} }}%
        =%
        O\pth{ \Bigl. \Ex{ \bigl. \cardin{\ObjX{\RSample}}} f(n/r) },
    \end{math}
    where the expectation is over the sample $\RSample$.
\end{theorem}

\subsection{Complexity of the overlay of %
   lower-envelopes of functions in \RIC}
\seclab{functions}

Let $\FuncSet$ be a set of $n$ functions, such that for all $f\in
\FuncSet$, we have %
\begin{inparaenum}[(1)]
    \item $f:\Re \rightarrow \Re$, and
    \item $f$ is continuous.
\end{inparaenum}
The \emphi{curve} associated with $f$ is its image $\brc{ (x, f(x))
   \sep{ x \in \Re}}$. We use $f$ to refer both to the function and
its curve.

We assume that any pair of curves in $\FuncSet$ only intersect
transversally and at most $\intersections$ times, and that no three
curves intersect at a common point (i.e. general position), where
$\intersections$ is some small constant.  Here
$\OFuncSet = \permut{f_1, \ldots, f_n}$ denotes a fixed permutation of
the $n$ functions, $\OFuncSet_i = \permut{f_1, \ldots, f_i}$ denotes a
prefix of this permutation, and $\FuncSet_i = \brc{f_1, \ldots, f_i }$
is the associated unordered set.

Let $m_i$ be the number of vertices (i.e. intersections of functions)
on the lower envelope of $\FuncSet_i$ that are not present in the
lower envelope of $\FuncSet_{i-1}$.  For a given permutation
$\OFuncSet$ of $\FuncSet$, we define the \emphi{overlay complexity} to
be the quantity $\overlayX{\OFuncSet} = \sum_{i=1}^n m_i$.  In other
words, when we insert the $i$\th function we create a number of new vertices on
the lower envelope of $\OFuncSet_i$.  If we shoot down a vertical ray from 
each such vertex when it is created, then
$\overlayX{\OFuncSet}$ is the number of distinct locations on the
$x$-axis that get hit by rays over the entire randomized incremental
construction of the lower-envelope.

Let $\dsFunc{\intersections}{y}$ denote the maximum length of a
Davenport-Schinzel sequence of order $\intersections$ on $y$ symbols.
The function $\dsFunc{\intersections}{y}$ is monotonically increasing,
and slightly super linear for $\intersections \geq 3$, for example
$\dsFunc{\intersections}{y} = O\pth{ y \cdot 2^{O\pth{
         (\alpha(y))^\intersections}}}$,
where $\alpha$ is the inverse Ackermann function (for the currently
best bounds known, see \cite{p-sbdss-13}). The conditions on the
functions in $\FuncSet$ give us the following.
\begin{observation}
    \obslab{ds}%
    For $i=1,\ldots, n$, the number of vertices on the lower envelope
    of $\OFuncSet_i$ is $O \pth{\dsFunc{\intersections}{i}}$, where
    $\intersections$ is a constant (which is determined by the number
    of times pairs of curves are allowed to intersect), see
    \cite{sa-dsstg-95}.
\end{observation}

\begin{lemma}
    \lemlab{nice:functions:l:e}%
    %
    Let $\OFuncSet = \permut{f_1, \ldots, f_n}$ be a random
    permutation of a set of continuous functions $\FuncSet$, where
    every pair of associated curves intersect at most $\intersections$
    times, where $\intersections$ is some constant.  Then
    $\Ex{\overlayX{\OFuncSet_n}} = O(\dsFunc{\intersections}{n})$.
\end{lemma}
\begin{proof}
    By definition we have that
    $\Ex{\bigl. \overlayX{\OFuncSet_n}} = \allowbreak
    \Ex{\bigl. \sum_{i=1}^n m_i} \allowbreak = \sum_{i=1}^n \Ex{m_i}$,
    where $m_i$ is the number of vertices on the lower envelope of
    $\FuncSet_i$ that are not present on the lower envelope of
    $\FuncSet_{i-1}$.  Consider a vertex, $v$, on the lower envelope
    of $\FuncSet_i$, for some $1\leq i\leq n$.  Let $X_v$ be an
    indicator variable which is 1 if and only if $v$ was not present
    in $\FuncSet_{i-1}$.  Since $\OFuncSet$ is a random permutation of
    $\FuncSet$, it holds that $\OFuncSet_i$ is a random permutation of
    $\FuncSet_i$.  Since any vertex on the lower envelope is defined
    by exactly two functions from $\FuncSet_i$, it holds that
    $\Ex{X_v} = 2/i$, since $X_v$ is $1$ if and only if one of $v$'s
    two defining functions was the last function, $f_i$, in the
    permutation $\OFuncSet_i$.  Therefore,
    \begin{align*}
        \Ex{\bigl. m_i}%
        &=%
        \Ex{\sum_{v\in \LE[]{\FuncSet_i}} X_v}%
        =%
        \sum_{v\in \LE[]{\FuncSet_i}} \Ex{\bigl. X_v}%
        =%
        \sum_{v\in \LE[]{ \FuncSet_i}} \frac{2}{i}%
        =%
        \frac{2 \cardin{\LE{\FuncSet_i}}}{i},
    \end{align*}
    where $\LE[]{\FuncSet_i}$ is the set of vertices on the lower
    envelope of $\FuncSet_i$. By \obsref{ds},
    $\cardin{\LE[]{\FuncSet_i}} = O\pth{ \dsFunc{\intersections}{i}}$.
    We thus have
    \begin{align*}
        \Ex{\bigl. \overlayX{\OFuncSet} }%
        =%
        \sum_{i=1}^n \Ex{ \bigl. m_i}%
        \leq%
        \sum_{i=1}^n O\pth{ \frac{\dsFunc{\intersections}{i}}{i} }%
        \leq%
        \sum_{i=1}^n O\pth{ \frac{\dsFunc{\intersections}{n}}{n} }%
        =%
        O\pth{\bigl. \dsFunc{\intersections}{n}},
    \end{align*}
    as $\dsFunc{\intersections}{i}/i$ is a monotonically increasing
    function \cite{sa-dsstg-95}. %
    \DCGVer{{\qed}}
\end{proof}
\begin{corollary}
    \corlab{intersections}%
    Let $\ell$ be a bisector defined by a pair of disjoint sites
    $\site_1$ and $\site_2$.  Let $\SiteSet$ be a set of $n$ sites
    containing $\site_1$ and $\site_2$ (and satisfying the conditions
    of \remref{assumptions}), and let
    $\OSiteSet=\permut{\site_1, \site_2, \ldots, \site_n}$ be a
    permutation of $\SiteSet$, such that
    $\OSuffix{3} = \permut{\site_3, \site_4, \ldots, \site_n} $ is a
    random permutation.  Finally, let $\VCellPrefix{i}$ denote the
    Voronoi cell of $\site_i$ in $\VorX{\OSiteSet_i}$.
    
    The expected number of intersection points of $\ell$ with the
    boundaries of
    $\VCellPrefix{3}, \VCellPrefix{4}, \ldots, \VCellPrefix{n}$, is
    $O(\dsFunc{\intersections}{n})$, for some constant
    $\intersections$.%
\end{corollary}
\begin{proof}
    Consider the distance between any site and a point on $\ell$.
    This distance can be viewed as a parameterized real valued
    function as we move along $\ell$.  For a given site $\site_i$ let
    us denote this function $\func_i(t)$ (where $t$ is the location
    along $\ell$).  Clearly such distance functions are continuous as
    we move along any curve, and in particular along $\ell$.  Consider
    a point $t$ where two functions intersect, i.e. $\func_i(t) =
    \func_j(t)$ for some $i\neq j$.  This corresponds to a point on
    the bisector of $\site_i$ and $\site_j$.  Since $\ell$ is a
    bisector and we assumed that any two bisectors intersect at most a
    constant number of times, for any fixed $i$ and $j$, there are at
    most a constant number of points along $\ell$ such that
    $\func_i(t) = \func_j(t)$.  Therefore, the functions $\func_i$
    representing the distance to site $\site_i$ satisfy the conditions
    to apply \lemref{nice:functions:l:e}.
    
    Consider a Voronoi edge on the boundary of some cell in
    $\VCellPrefix{3}, \ldots, \VCellPrefix{n}$ which crosses $\ell$.
    Each such edge is defined by a subset of the bisector of two
    sites, and let these sites be $\site_i$ and $\site_j$ where $i<j$.
    We are interested at the point when the edge crosses $\ell$, and
    therefore this corresponds to a point $t$ on $\ell$ such that
    $\func_i(t) = \func_j(t)$.  Moreover, in order for this edge to
    appear on the boundary of $\VCellPrefix{j}$ we have that
    $\func_i(t) = \func_j(t) < \func_k(t)$ for all
    \begin{math}
        k \in \brc{1, \ldots, j } \setminus \brc{ i, j}.
    \end{math}
    In other words, the point where $\func_i(t) = \func_j(t)$ must
    appear on the lower envelope of $\func_1(t), \ldots, \func_j(t)$.
    Therefore, in order to bound the total expected number of
    intersection points of edges with $\ell$, it suffices to bound the
    total expected number of vertices ever seen on the lower envelope
    of these functions when inserting the sites in a random order
    $\OSuffix{3}$ (note that one also has to factor in the complexity
    of the lower envelope due to $\site_1$ and $\site_2$, but this
    only contributes a constant factor blow up).  The result now
    readily follows from \lemref{nice:functions:l:e}. %
    \DCGVer{{\qed}}
\end{proof}

\subsubsection{Bounding the overlay complexity of %
   Voronoi cells of sites}

The following lemma uses an interesting backward-forward analysis that
the authors had not encountered before, and might be of independent
interest.

\begin{lemma}
    \lemlab{overlay:2}%
    %
    Let $\OSiteSet = \permut{\site_1, \ldots, \site_n}$ be a random
    permutation of a set $\SiteSet$ of sites in the plane, complying
    with the conditions of \remref{assumptions} and
    \remref{assumptions:2}.  Let $\VCellPrefix{i}$ denote the Voronoi
    cell of $\site_i$ in $\VorX{\OSiteSet_i}$.  The expected total
    complexity of the overlay arrangement $\Arr = \Arr\pth{
       \VCellPrefix{1}, \ldots, \VCellPrefix{n}}$ is $
    O(\dsFunc{\intersections}{n} \log n)$, for some constant
    $\intersections$.
\end{lemma}

\begin{proof}
    As discussed in the beginning of \secref{compat}, in order to
    bound $\cardin{\Arr}$ it suffices to bound the number of vertices
    in the arrangement.  By planarity (and since there are no isolated
    vertices) it also suffices to bound the number of edges.
    
    Let $\arcsX{\VCellPrefix{ i }}$ be the Voronoi edges in $\VorX{
       \OSiteSet_i}$ that appear on the boundary of $\VCellPrefix{i}$.
    Such an arc $\bisector \in \arcsX{\VCellPrefix{ i }}$, created in
    the $i$\th iteration, is going to be broken into several edges in
    the final overlay arrangement $\Arr$.  Let $Z_\bisector$ be the
    number of such edges that arise from $\bisector$.  Our goal is to
    bound the quantity $\Ex{\sum_i \sum_{\bisector \,\in\, \arcsX[]{
             \VCellPrefix{i} }} Z_\bisector}$.

    \smallskip%
    Each Voronoi edge, $\edge$, in the Voronoi diagram of a subset of
    the sites, is defined by a constant number of sites (the two sites
    whose bisector it is on, and the two sites that delimit it), and
    it has an associated stopping set. The stopping set (i.e.,
    conflict list), $\KillSet{\edge}$, is the set of all sites whose
    insertion prevents $\edge$ from appearing in the Voronoi diagram in
    its entirety.
    
    For the rest of the proof we fix the prefix $\SiteSet_i$; that is,
    fix the sites that are the first $i$ sites in the permutation
    $\OSiteSet$, but not their internal ordering in the
    permutation. Naturally, this also determines the content of the
    suffix $\Suffix{i+1} = \SiteSet \setminus \SiteSet_i$.  Consider
    an edge, $\edge$, which lies on a bisector defined by sites
    $\site_j$ and $\site_i$, where $j<i$.  Then since $\OSuffix{i+1}$
    is a random permutation of $\Suffix{i+1}$, \corref{intersections}
    implies that
    \begin{math}
        \Ex{\bigl.  Z_\edge } = O\pth{ \Bigl. \dsFunc{\intersections}{
              \bigl.\cardin{\KillSet{\edge}}}},
    \end{math}
    where $\intersections$ is some constant, and the expectation is
    over the internal ordering of $\OSuffix{i+1}$.
    
    For an edge $\edge \in \VorX{\SiteSet_i}$, let $X_\edge$ be an
    indicator variable that is one if $\edge$ was created in the
    $i$\th iteration, and furthermore, it lies on the boundary of
    $\VCellPrefix{i}$.  Observe that $\Ex{\bigl.  X_\edge} \leq 4/i$,
    as an edge appears for the first time in round $i$ only if one of
    its (at most) four defining sites was the $i$\th site inserted.
    
    Let
    \begin{math}
        Y_i%
        =%
        \sum_{\bisector \,\in\, \arcsX[]{ \VCellPrefix{i} }}
        Z_\bisector%
        =%
        \sum_{\edge \in \VorX[]{\SiteSet_i}} Z_ \edge X_\edge
    \end{math}
    be the total (forward looking) complexity contribution to the
    final arrangement $\Arr$ of arcs added in round $i$.  As we
    assumed $\SiteSet_i$ is fixed, hence correspondingly
    $\Suffix{i+1}$ is fixed.  Let $\edge$ be some edge in
    $\VorX[]{\SiteSet_i}$.  Observe that the value $Z_\edge$ depends
    only on the internal ordering $\OSuffix{i+1}$ of the suffix
    $\Suffix{i+1}$, and the indicator variable $X_\edge$ depends only
    on the internal ordering $\OSiteSet_i$ of the prefix $\SegSet_i$.
    In other words, for a fixed $\SiteSet_i$ and edge $\edge$ in
    $\VorX[]{\SiteSet_i}$, the random variables $Z_\edge$ and
    $X_\edge$ are independent.  We thus have
    \begin{align*}
        \Ex{Y_i\sep{ \SiteSet_i}\!}%
        &=%
        \Ex{ \sum\nolimits_{\edge \in \VorX[]{\SiteSet_i}} Z_\edge
           X_\edge \sep{\Bigl.  \SiteSet_i}}%
        =%
        \sum_{\edge \in \VorX[]{\SiteSet_i}}%
        \Ex{Z_\edge \sep{ \SiteSet_i } } \Ex{X_\edge \sep{
              \SiteSet_i}}%
        \DCGVer{\\&}%
        =%
        \sum_{\edge \in \VorX[]{\SiteSet_i}}%
        O\pth{\Bigl. \dsFunc{\intersections}{
              \cardin{\KillSet{\edge}}\bigr.}}  \Ex{X_\edge \sep{
              \SiteSet_i}}%
        \\&=%
        O\pth{\frac{1}{i}\sum_{\edge \in \VorX[]{\SiteSet_i}}
           \dsFunc{\intersections}{\Bigl. \cardin{\KillSet{\edge }}}}.
    \end{align*}

    We can now get a bound on the expected value of $Y_i$, as we have
    a bound for this quantity when conditioned on $\SiteSet_i$, as
    $\Ex{Y_i\bigr.} = \Ex{ \bigl. \Ex{Y_i \mid { \SiteSet_i}}}\Bigr.$.
    Specifically, we will apply the Clarkson-Shor technique, described
    in \secref{clarkson}, where the set of elements is the set of
    sites $\SiteSet$, the prefix $\SiteSet_i$ is the random sample,
    and the edges of $\VorX[]{\SiteSet_i}$ form the set of defined
    objects. Since the complexity of an unweighted Voronoi diagram of
    sites is always linear, the Clarkson-Shor technique (i.e.,
    \thmref{moments}) implies
    \begin{math}
        \nu_i%
        =%
        \Ex{\sum_{\edge \in \VorX[]{\SiteSet_i}}
           \dsFunc{\intersections}{\bigl.\cardin{\KillSet{\edge}}} }
        =%
        O\pth{\bigl. \Ex{\bigl.\cardin{\VorX[]{\SiteSet_i}}}
           \dsFunc{\intersections}{n/i}}= O\pth{\bigl. i \,
           \dsFunc{\intersections}{n/i}},
    \end{math}
    where the randomness here is on the choice of the sites that are
    in the $i$\th prefix $\SiteSet_i$.
    
    The total complexity of $\Arr$ is asymptotically bounded by $\sum_i
    Y_i$, and we have
    \begin{align*}
        \Ex{\sum_i Y_i}%
        &=%
        \sum_i \Ex{ \bigl. Y_i}%
        =%
        \sum_i \Ex{\Bigl.\Ex{Y_i\sep{ \SiteSet_i}}}%
        \DCGVer{\\&}%
        =%
        \sum_i O\pth{ %
           \frac{1}{i}%
           \ExChar \Bigl[\,\,\! {%
              \sum_{\edge \in \VorX[]{\SiteSet_i}}
              \dsFunc{\intersections}{\bigl.  \cardin{\KillSet{\edge
                    }}%
              } } \Bigr]%
        }%
        \displaybreak[0] \\%
        &=%
        O\pth{\sum_i \frac{1}{i} \nu_i }%
        =%
        O\pth{\sum_i \dsFunc{\intersections}{n/i} }%
        =%
        O\pth{\sum_i \frac{\dsFunc{\intersections}{n}}{i}}%
        \DCGVer{\\&}%
        =%
        O\pth{ \bigl. \dsFunc{\intersections}{n} \log n}.%
        \DCGVer{\tag*{\qed}}
    \end{align*} %
\end{proof}

\section{The Result and Applications}
\seclab{resultApplications}

We now consider the various applications of our technique.  In
\thmref{points} it was already observed that a bound of $O(n \log^3
n)$ holds on the expected complexity of the weighted Voronoi diagram
when the sites are points.  We can now extend this result to more
general sites by combining \thmref{sitesTheorem} and
\lemref{overlay:2}.  We first present this more general result, with a
slightly tightened analysis (specifically a $\log$ factor
improvement), and then describe the applications of this result.

\subsection{The result}
\seclab{tighter}

\begin{theorem}
    \thmlab{tighter}%
    %
    Let $\SiteSet$ be a set of $n$ sites in the plane, satisfying the
    conditions of \remref{assumptions} and \remref{assumptions:2},
    where for each site we independently sample a weight from some
    distribution $\pdf$ over $\Re^+$.  Then the expected complexity of
    the multiplicative Voronoi diagram of $\SiteSet$ is
    $O\pth{\bigl.  \dsFunc{\intersections}{n}\, \worstX{\log n} }$.
\end{theorem}
\begin{proof}
    Adopting previously used notation, let $\SiteSet$ be a randomly
    weighted set of sites in the plane, whose ordering by increasing
    weight is $\OSiteSet = \permut{\site_1, \ldots, \site_n}$, and let
    $\VCellPrefix{i}$ denote the Voronoi cell of $\site_i$ in the
    unweighted Voronoi diagram of $\SiteSet_i$. Let $\Arr_i$ denote
    the overlay arrangement of the regions $\VCellPrefix{1}, \ldots,
    \VCellPrefix{i}$.  Now, $\OSiteSet_i = \permut{\site_1, \ldots,
       \site_i}$ is a random permutation of $\SiteSet_i$, and
    \lemref{overlay:2} implies that the expected complexity of
    $\Arr_i$ is $O(\dsFunc{\intersections}{i} \log i)$ for any $i\leq
    n$.
    
    Consider the arrangement $\Arr_{n/t}$, determine by the first
    $n/t$ sites, where $t$ is parameter to be determined shortly.
    Just as in the proof of \thmref{sitesTheorem}, consider the
    arrangement $\ArrVDX{n/t}$ formed by vertical decomposition of
    $\Arr_{n/t}$. The vertical decomposition increases the complexity
    only by a constant factor, and thus the expected number of
    vertical trapezoids is
    $O(\dsFunc{\intersections}{n/t} \log (n/t))$ (where the
    expectation is over the ordering $\OSiteSet_{n/t}$ of
    $\SiteSet_{n/t}$).  Moreover, each cell (i.e., vertical trapezoid)
    is defined by a constant number of sites from $\SiteSet$ --
    specifically, a site is in the stopping set of a trapezoid if when
    added to the sample its Voronoi cell intersects the trapezoid.
    
    So consider a cell $\Delta$ in the arrangement $\ArrVDX{n/t}$.  By
    \lemref{constant}, with respect to the set $\SiteSet_{n/t}$, all
    points in $\Delta$ have the same candidate set.  However, as sites
    in $\Suffix{n/t+1}$ are added candidate sets of different points
    in $\Delta$ may diverge.  Clearly this can only happen when for
    some $j>n/t$, $\VCellPrefix{j}$ intersect $\Delta$, in other
    words, when $\site_j$ is in the stopping set $\KillSet{\Delta}$ of
    $\Delta$.

    Therefore, the union of the final candidate sets over all points
    in $\Delta$ has size $O(\KillSet{\Delta}+\log n)$, since all
    points had the same candidate set with respect to $\SiteSet_{n/t}$
    (which has size $O(\log n)$ by \corref{forAll}), and can only
    differ on the set $\KillSet{\Delta}$.  Since the worst case
    complexity of a weighted Voronoi diagram of $m$ sites is
    $\worstX{m}$, this implies the total complexity of the weighted
    Voronoi diagram in the cell $\Delta$, formed by the candidate list
    and stopping set of $\Delta$, is
    $O\pth{\Bigl.\worst\pth{\bigl.\cardin{ \KillSet{\Delta}}+\log
          n}}$.  Now we can apply \thmref{moments} to bound the sum of
    this quantity over all cells in the vertical decomposition of
    $\ArrVDX{n/t}$. Specifically, setting $t=\log n$, we have
    \begin{align*}
        \Ex{\Bigl.\,\smash{\sum_{\,\Delta\in \ArrVDX{n/t}}}
           \worstX{\bigl.  \cardin{\KillSet{\Delta}}+\log n } }%
        &%
        =%
        O\pth{%
           \Biggl.  \Ex{ \Bigl.\bigl| \ArrVDX{n/t} \bigr|}
           \worstX{\bigl. t+\log n}%
        }%
        \DCGVer{\\&}%
        =%
        O\pth{\Bigl.\dsFunc{\intersections}{\frac{n}{t}} \log
           \pth{\frac{n}{t}}
           \worstX{\bigl. t+\log n}    }\\
        &=%
        O\pth{\Bigl.\dsFunc{\intersections}{\frac{n}{\log n}} \log
           \pth{n} \, \worstX{\bigl. \log n} }%
        \DCGVer{\\&}%
        =%
        O\pth{\Bigl. \dsFunc{\intersections}{n } \worstX{\bigl. \log
              n} },
    \end{align*}
    as $\worstX{m}$ is a polynomially growing function, and using
    $\dsFunc{\intersections}{n/t} \leq \dsFunc{\intersections}{n}/t$.
    \DCGVer{{\qed}}
\end{proof}

\begin{corollary}
    \corlab{pointsCor}%
    Let $\PntSet$ be a set of $n$ points in the plane, where for each
    point we independently sample a weight from some distribution
    $\pdf$.  Then, the expected complexity of the multiplicative
    Voronoi diagram of $\PntSet$ is $O\pth{ n\log^2 n }$.
\end{corollary}

\subsubsection{Sampling versus Permutation}
\seclab{permutation}

The arguments used throughout this paper did not require that weights
were randomly sampled, but rather that they were randomly permuted.  A
similar observation was made by Agarwal
\etal~\cite{ahks-urmsn-14}. Specifically, we have the following
analogous lemma to \corref{pointsCor} (a similar lemma holds for more
general sites).

\begin{lemma}
    Let $W = \brc{\weight_1, \ldots, \weight_n}$ be a set of
    non-negative real weights and
    $\PntSet = \brc{ \pnt_1, \ldots, \pnt_n }$ a set of points in the
    plane.  Let $\sigma$ be a (uniformly) random permutation from the
    set of permutations on $\brc{1, \ldots, n }$.  If for all $i$ we
    assign $\weight_{\sigma(i)}$ to point $\pnt_i$, then the expected
    complexity of the resulting multiplicative Voronoi diagram of
    $\PntSet$ is $O\pth{ n\log^2 n }$.
\end{lemma}

\subsubsection{If the locations are sampled}
\seclab{locations}

Consider the alternative problem where one is given a set of points
with fixed weights and one then randomly samples the location of each
point.  It is not hard to see that this is equivalent to first
randomly sampling locations of points, and then randomly permuting the
weights among the locations.  This implies the following corollary.

\begin{corollary}
    \corlab{square}%
    Let $\PntSet = \brc{\pnt_1, \ldots, \pnt_n}$ be a set of points
    with an associated set of weights
    $W = \brc{\weight_1, \ldots, \weight_n}$ such that
    $\weight(\pnt_i)=\weight_i$.  If for all $i$ one picks the
    location of $\pnt_i$ uniformly at random from the unit square,
    then the expected complexity of the multiplicative Voronoi diagram
    is $O\pth{ n\log^2 n }$.
\end{corollary}

\begin{remark}
    It is likely that one can improve the bound in \corref{square}.
    Specifically, we are not using that the locations are sampled, but
    merely that the weights are permuted across the points.  In
    particular, for this special case it is likely one can improve the
    bound of Kaplan \etal \cite{krs-omdri-11} for the overlay
    complexity of the unweighted cells.
\end{remark}

\subsection{Applications}

For the following applications of \thmref{tighter}, the work of Sharir
\cite{s-atubl-94} implies the bound
$\worstX{ m } = O\pth{ m^{2+\eps} }$.

\subsubsection{Disjoint Segments}

Let $\SegSet$ be a set of $n$ interior disjoint line segments in the
plane.  The bisector of any two interior disjoint segments in the
plane consists of at most a constant number of pieces, where each
piece is a contiguous part of either a line or parabolic curve.  It is
therefore not hard to argue that $\SegSet$ satisfies all the
requirements on sets of sites from \remref{assumptions} and
\remref{assumptions:2}.

\begin{theorem}
    \thmlab{segMain}%
    Let $\SegSet$ be a set of $n$ interior disjoint segments in the
    plane, where for each segment we independently sample a weight
    from some distribution $\pdf$.  Then, the expected complexity of
    the multiplicative Voronoi diagram of $\SegSet$ is $O\pth{ n
       \log^{2+\eps} n}$.
\end{theorem}

Interpreting the Voronoi diagram as a minimization diagram, taking a
level set corresponds to taking the union of a randomly expanded set
of segments. Therefore, our bound immediately implies a bound of
$O\pth{ n \log^{2+\eps} n }$ on the complexity of the union of such
segments. Recently, Agarwal \etal~\cite{ahks-urmsn-14} proved a better
bound of $O(n \log n)$, but arguably our proof is significantly
simpler.

\subsubsection{Convex Sets}
Let $\mathsf{C}$ be a set of $n$ disjoint convex constant complexity
sets in the plane.  Note this is a clear generalization of the case of
segments, and for this case it is again not hard to verify that such a
set of sites meet all the requirements of \remref{assumptions} and
\remref{assumptions:2}.

\begin{theorem}
    \thmlab{convexMain} Let $\mathsf{C}$ be a set of $n$ interior
    disjoint convex constant complexity sets in the plane, where for
    each set we independently sample a weight from some distribution
    $\pdf$.  Then, the expected complexity of the multiplicative
    Voronoi diagram of $\mathsf{C}$ is $O\pth{ n \log^{2+\eps} n}$.
\end{theorem}

Again interpreting the Voronoi diagram as a minimization diagram, this
immediately implies a bound of $O\pth{ n \log^{2+\eps} n }$ on the
complexity of the union of a set of such randomly expanded convex
sets. Agarwal \etal~\cite{ahks-urmsn-14} proved a bound of
$O(n^{1+\eps})$ for any fixed $\eps>0$, and as such the above bound is
an improvement.


\section{Conclusions}

In this paper, we presented a general technique to provide an expected
near linear bound on the combinatorial complexity of a large class of
multiplicative Voronoi diagrams, when the weights are sampled
randomly, which have quadratic complexity in the worst case.  Several
specific applications of the technique were listed, but there should
probably be more of such applications.  There is also some potential
to improve the bounds in the paper.  For example, one can likely use
the uniform distribution of the points to improve the result in
\corref{square}.  However, we conjecture that, in the worst case, 
the expected complexity should still be super linear, 
and we provide some justification for this 
conjecture in \apndref{lower:bound}.

In order to achieve our bounds we introduced the notation of candidate
sets, which induce a planar partition into uniform candidate regions.
Recently, we considered this partition as a diagram of independent
interest \cite{chr-fpuvp-14}.  Generalizing to allow each site to have
multiple weights, this one diagram captures the relevant information
for multi-objective optimization, i.e. this one diagram implies bounds
on various weighted generalizations of Voronoi diagrams.  Moreover, by
extending the techniques of the current paper, we provide a similar
bounds on the expected complexity of this diagram \cite{chr-fpuvp-14}.



\paragraph{Acknowledgments.}

The authors would like to thank Pankaj Agarwal, Jeff Erickson, Haim
Kaplan, Hsien-Chih Chang, and Micha Sharir for useful discussions. In
particular, the work of Agarwal, Kaplan, and Sharir
\cite{aks-urmsn-13,ahks-urmsn-14} was the catalyst for this work. In
addition, we thank Pankaj Agarwal for pointing out a simple way to
slightly improve our bound, specifically the result in
\thmref{tighter}.  The authors would also like to thank the reviewers
for their insightful comments.
\DCGVer{%
   Work on this paper was partially supported by NSF AF awards
   CCF-0915984, CCF-1217462, and CCF-1421231. 
   A preliminary version of this paper appeared in SoCG 2014
   \cite{hr-ocrwv-14}.%
}


\DCGVer{%
 \providecommand{\CNFX}[1]{ {\em{\textrm{(#1)}}}}
  \providecommand{\tildegen}{{\protect\raisebox{-0.1cm}{\symbol{'176}\hspace{-0.03cm}}}}
  \providecommand{\SarielWWWPapersAddr}{http://sarielhp.org/p/}
  \providecommand{\SarielWWWPapers}{http://sarielhp.org/p/}
  \providecommand{\urlSarielPaper}[1]{\href{\SarielWWWPapersAddr/#1}{\SarielWWWPapers{}/#1}}
  \providecommand{\Badoiu}{B\u{a}doiu}
  \providecommand{\Barany}{B{\'a}r{\'a}ny}
  \providecommand{\Bronimman}{Br{\"o}nnimann}  \providecommand{\Erdos}{Erd{\H
  o}s}  \providecommand{\Gartner}{G{\"a}rtner}
  \providecommand{\Matousek}{Matou{\v s}ek}
  \providecommand{\Merigot}{M{\'{}e}rigot}
  \providecommand{\CNFSoCG}{\CNFX{SoCG}}
  \providecommand{\CNFCCCG}{\CNFX{CCCG}}
  \providecommand{\CNFFOCS}{\CNFX{FOCS}}
  \providecommand{\CNFSODA}{\CNFX{SODA}}
  \providecommand{\CNFSTOC}{\CNFX{STOC}}
  \providecommand{\CNFBROADNETS}{\CNFX{BROADNETS}}
  \providecommand{\CNFESA}{\CNFX{ESA}}
  \providecommand{\CNFFSTTCS}{\CNFX{FSTTCS}}
  \providecommand{\CNFIJCAI}{\CNFX{IJCAI}}
  \providecommand{\CNFINFOCOM}{\CNFX{INFOCOM}}
  \providecommand{\CNFIPCO}{\CNFX{IPCO}}
  \providecommand{\CNFISAAC}{\CNFX{ISAAC}}
  \providecommand{\CNFLICS}{\CNFX{LICS}}
  \providecommand{\CNFPODS}{\CNFX{PODS}}
  \providecommand{\CNFSWAT}{\CNFX{SWAT}}
  \providecommand{\CNFWADS}{\CNFX{WADS}}

}

\NotDCGVer{%
 \providecommand{\CNFX}[1]{ {\em{\textrm{(#1)}}}}
  \providecommand{\tildegen}{{\protect\raisebox{-0.1cm}{\symbol{'176}\hspace{-0.03cm}}}}
  \providecommand{\SarielWWWPapersAddr}{http://sarielhp.org/p/}
  \providecommand{\SarielWWWPapers}{http://sarielhp.org/p/}
  \providecommand{\urlSarielPaper}[1]{\href{\SarielWWWPapersAddr/#1}{\SarielWWWPapers{}/#1}}
  \providecommand{\Badoiu}{B\u{a}doiu}
  \providecommand{\Barany}{B{\'a}r{\'a}ny}
  \providecommand{\Bronimman}{Br{\"o}nnimann}  \providecommand{\Erdos}{Erd{\H
  o}s}  \providecommand{\Gartner}{G{\"a}rtner}
  \providecommand{\Matousek}{Matou{\v s}ek}
  \providecommand{\Merigot}{M{\'{}e}rigot}
  \providecommand{\CNFSoCG}{\CNFX{SoCG}}
  \providecommand{\CNFCCCG}{\CNFX{CCCG}}
  \providecommand{\CNFFOCS}{\CNFX{FOCS}}
  \providecommand{\CNFSODA}{\CNFX{SODA}}
  \providecommand{\CNFSTOC}{\CNFX{STOC}}
  \providecommand{\CNFBROADNETS}{\CNFX{BROADNETS}}
  \providecommand{\CNFESA}{\CNFX{ESA}}
  \providecommand{\CNFFSTTCS}{\CNFX{FSTTCS}}
  \providecommand{\CNFIJCAI}{\CNFX{IJCAI}}
  \providecommand{\CNFINFOCOM}{\CNFX{INFOCOM}}
  \providecommand{\CNFIPCO}{\CNFX{IPCO}}
  \providecommand{\CNFISAAC}{\CNFX{ISAAC}}
  \providecommand{\CNFLICS}{\CNFX{LICS}}
  \providecommand{\CNFPODS}{\CNFX{PODS}}
  \providecommand{\CNFSWAT}{\CNFX{SWAT}}
  \providecommand{\CNFWADS}{\CNFX{WADS}}

}


\appendix

\section{Lower bound on the overlay complexity of %
   Voronoi cells in \RIC}
\apndlab{lower:bound}

Kaplan \etal \cite{krs-omdri-11} provided an example showing that in
the randomized incremental construction of the lower envelope of
planes in $3d$, the overlay of the cells being computed in the
minimization diagram has expected complexity $\Omega( n \log
n)$. Their example however is not realizable by a Voronoi
diagram. Here we provide a direct example showing the $\Omega(n \log
n)$ lower bound for the overlay of Voronoi cells in the randomized
incremental construction.

Because of the following lemma, we conjecture that, in the worst case,
the true complexity of the quantity bounded in \thmref{points} is
super linear. We leave this as open problem for further research.

\begin{lemma}
    For $n$ sufficiently large, there is a set of $2n$ points in the
    plane such that the overlay of the Voronoi cells computed in the
    randomized incremental construction of the Voronoi diagram has
    expected complexity $\Omega(n \log n)$.
\end{lemma}
\begin{proof}
    Let $\PntSet$ be a set of $2n$ points, where the $i$\th point is
    $\pnt_i = (i, -\Delta)$ and the $(n+i)$\th point is $\pntA_i = (i
    , +\Delta)$, for $i=1,\ldots, n$, where $\Delta$ is a sufficiently
    large number, say $10n^3$. Let $\OSiteSet =\permut{\site_1,
       \ldots, \site_{2n}}$ be a random permutation of the points of
    $\PntSet$, and let
    \begin{align*}
        \VCellPrefix{i} = \VorCell{\site_i}{\OSiteSet_i},
    \end{align*}
    for $i=1,\ldots, 2n$.
    
    \begin{figure}
        \centerline{%
           \begin{tabular}{cc}
               \includegraphics[page=1]{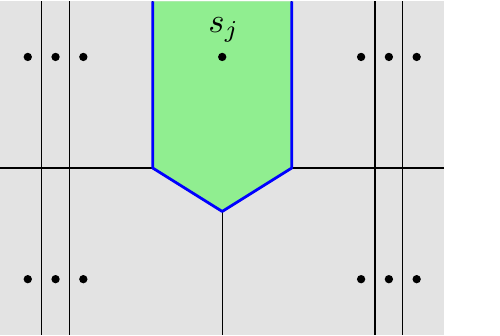} &
               \includegraphics[page=2]{figs/isolated}\\
               (A) & (B)
           \end{tabular}%
        }%
        \caption{(A) The site $\site_j$ is isolated. (B) The overlay
           vertices of the cell $\partial \VCellPrefix{j}$ with the
           boundary of cells created later.  Note the figure is
           vertically not to scale.}
        \figlab{isolated:2}
    \end{figure}
    
    Let $\beta = 10 \ceil{\lg n}$, and let $\Event$ be the event that
    in the first $\beta$ sites, there are sites that belong to both
    the top and bottom row. We have that $\rho = \Prob{\bigl. \Event}
    = 1 - 2\binom{n}{\beta}/\binom{2n}{\beta} \geq 1-2/n^{10}$, as
    $2^\beta \binom{n}{\beta} \leq \binom{2n}{\beta}$.  The $j$\th
    site $\site_j$ (say it is located at $(x_j, \Delta)$) is
    \emphi{isolated}, if none of the points $(x_j-\xi_j, \pm \Delta),
    (x_j -\xi_j +1, \pm \Delta) \ldots, (x_j+\xi_j, \pm \Delta) $ are
    present in the prefix $\OSiteSet_{j-1} = \permut{\site_1, \ldots,
       \site_{j-1}}$, where $\xi_j = \ceil{ n/8j}$.  If a site
    $\site_j$ is isolated, for $j \geq \beta$, then its cell is going
    to be U shaped (assuming $\Event$ happened), ``biting'' a portion
    of the $x$-axis, as $\Delta \gg n$, see \figref{isolated:2} (A).

    The probability of the site $\site_j$ inserted in the $j$\th
    iteration to be isolated is at least a half, since the majority of
    the points not inserted yet are isolated.  Indeed, consider a site
    $i$, for $i < j$, and consider the interval it ``blocks''
    $Z_i = [x_i - \xi_j, x_i + \xi_j]$ from being isolated. That is,
    if $x_j \in Z_i$, then $\site_j$ is not isolated. The total number
    of integer numbers in the intervals $Z_1,\ldots, Z_{j-1}$ is at
    most $\alpha_j = \pth{2\xi_j + 1} (j-1)$, and as such the first
    $j-1$ sites, block at most $2\alpha_j$ sites (that are located
    either on the top or bottom row) from being isolated in the $j$\th
    iteration. As such, we have
    \begin{align*}
        \Prob{\Bigl. \site_j \text{ is not isolated}}%
        &\leq%
        \frac{2\alpha_j}{2n - j}%
        =%
        \frac{2\pth{2\xi_j + 1} (j-1)}{2n-j}%
        \leq%
        \frac{2\pth{\bigl.2\ceil{ \nfrac{n}{8j}} + 1} j}{2n-j}%
        \\&%
        \leq%
        \frac{\pth{ \nfrac{n}{2j} + 6} j}{2n-j}%
        \leq%
        \frac{ {n}/{2} + 6 j}{2n-j}%
        \leq%
        \frac{ (1/2 + 6/20)n }{(2-1/20)n} %
        =%
        \frac{ 16 }{39} %
        \leq%
        \frac{1}{2},
    \end{align*}
    for $j \leq n/20$, and for $n$ sufficiently large.
    
    If $\site_j$ is indeed isolated (and we remind the reader that we
    assume it is located at $(x_j, \Delta)$), then there are no other
    sites (at this stage) in the slab $[x_j - \xi_j, x_j + \xi_j]
    \times [-\infty, +\infty]$. In particular, the interval $I_j =
    \big[ x_j - \xi_j/2, x_j + \xi_j/2\big]$ that lies on the $x$-axis
    is in the interior of the Voronoi cell $\VCellPrefix{j}$ of
    $\site_j$.
    
    This implies that in the final overlay arrangement,
    $\VCellPrefix{j}$ intersects all the cells of the sites
    $\pnt_{x_j-\xi_j/2}$, $ \ldots, \pnt_{x_j+\xi_j/2}$, as their
    cells intersect the interval $I_j$. This in turn implies that
    $\partial \VCellPrefix{j}$ contains at least $2\floor{\xi_j/2}$
    intersection with the boundaries of other cells in the final
    overlay, see \figref{isolated:2} (B). This counts only ``future''
    intersections of $\partial \VCellPrefix{j}$ with the boundaries of
    cells created later. In addition, a tiny perturbation in the
    locations of the sites, guarantees that the boundary of
    $\VCellPrefix{j}$, does not lie on the boundary of any other
    Voronoi cell being created in this process.  As such, an overlay
    vertex is being counted only once by this argument.
    
    We conclude that the expected complexity of the overlay is
    \begin{align*}
        &\geq%
        \Prob{\bigl. \Event} \sum_{j=\beta+1}^{n/20}
        2\floor{\bigl. \xi_j/2}%
        \Prob{\bigl. \site_j \text{ is isolated}} %
        \geq%
        \frac{1}{2} \cdot \frac{1}{2} \sum_{j=\beta+1}^{n/20}
        2\floor{\Bigl. \ceil{ \bigl. n/8j}/2}%
        \\&%
        \geq%
        \frac{1}{2} \sum_{j=\beta+1}^{n/20} \floor{\Bigl.  n/16j}%
        \geq%
        \frac{n}{64} \sum_{j=\beta+1}^{n/20} \frac{1}{j}%
        =%
        \Omega\pth{ \bigl. n \log n }. %
        \DCGVer{\tag*{\qed}}
    \end{align*}
\end{proof}


\end{document}